\documentclass[12pt,runningheads,a4paper]{llncs}

\usepackage{comment}
\usepackage{amssymb}
\usepackage{graphicx}
\usepackage{makeidx}
\usepackage{amsmath}
\usepackage{amsfonts}
\usepackage{amssymb}
\usepackage{xcolor}
\usepackage{gastex}
\usepackage{rotating}
\usepackage{slashbox}
\usepackage[misc,geometry]{ifsym}
\usepackage{geometry}
\long\def\remove#1{}

\newcommand{\pbox}{\hbox to 6pt{\leaders\hrule width 6pt height 6pt\hfill}}

\begin{document}
	
	\mainmatter  
	
	\title{The Complexity and Expressive Power of Second-Order Extended Logic}
	
	\titlerunning{The Complexity and Expressive Power of Second-Order Extended Logic}
	
	%
	%
	\author{Shiguang Feng \inst{1}
		\and Xishun Zhao \inst{1}
	}
	\authorrunning{S. Feng, X. Zhao}
	
	\institute{Institute of Logic and Cognition, Department of Philosophy, \\Sun Yat-sen University, Guangzhou, 510275, China.
	}
	
	%
	%
	
	\maketitle

\begin{abstract}
We study the expressive powers of SO-HORN$^{*}$, SO-HORN$^{r}$ and
SO-HORN$^{*r}$ on all finite structures. We show that SO-HORN$^{r}$,
SO-HORN$^{*r}$, FO(LFP) coincide with each other and SO-HORN$^{*}$
is proper sublogic of  SO-HORN$^{r}$.
To prove this result, we introduce the notions of DATALOG$^{*}$ program,  DATALOG$^{r}$ program and their stratified versions, S-DATALOG$^{*}$
program and S-DATALOG$^{r}$ program. It is shown that, on all structures,
DATALOG$^{r}$ and S-DATALOG$^{r}$ are equivalent and DATALOG$^{*}$
is a proper sublogic of DATALOG$^{r}$. SO-HORN$^{*}$ and SO-HORN$^{r}$
can be treated as the negations of DATALOG$^{*}$ and DATALOG$^{r}$,
respectively.  We also show that SO-EHORN$^{r}$ logic which is an extended version
of SO-HORN captures co-NP on all finite structures.
\end{abstract}

\section{Introduction}

Descriptive complexity is a bridge between complexity theory and mathematical
logic. It uses logic systems to measure the resources that are necessary
for some complexity classes. The first capture result is Fagin's theorem.
In 1974, Fagin showed that existential second-order logic($\exists$SO)
captures NP on all finite structures. This is a seminal work that
has been followed by many studies in the characterization of other
complexity classes by means of logical systems, the research of NP-complete
approximation problems from a descriptive point of view, and so on.
Through the effort of many researchers in this area, almost every
important complexity classes have their corresponding capture logic
systems. First-order logic (FO) which is a very powerful logic in classic
model theory is too weak when only consider finite structures. By augmenting
some recursion operators, researchers got many powerful logic. FO(IFP),
FO(PFP), FO(DTC) and FO(TC) are the logic that equip first-order logic
with inflationary fixpoint operator, partial fixpoint operator, deterministic
transitive closure operator and transitive closure operator, respectively.
Immerman and Vardi both showed that FO(IFP) captures PTIME on ordered
structures in 1986 and 1982, respectively. Abiteboul and Vainu showed
that FO(PFP) captures PSPACE on ordered structures in 1983. Immerman
showed that FO(DTC) and FO(TC) capture L and NL on ordered structures, respectively,
 in 1987. In\cite{gra91}, Gr\"{a}del introduced  SO-HORN logic
which is a fragment of second-order logic and showed that SO-HORN
captures PTIME on ordered structures. Whereas all these results are
obtained on ordered structures, a very important open problem  in
descriptive complexity is that whether there is an effective logic
captures PTIME on all structures. It turns out that if there is no
such a logic, then $P\neq NP$. Due to the definition of SO-HORN,
on all structures, it is closed under substructures and can not express
some first-order definable properties. Gr\"{a}del introduced another version
of this logic in \cite{Gradel} which is called SO-HORN$^{*}$ in this paper.
SO-HORN$^{*}$ allows first-order formula in it, but we show that it still
is not expressive enough to capture PTIME on all structures. In\cite{zhao-1},
we defined SO-HORN$^{r}$ logic which is an revised version of SO-HORN
and showed that it captures PTIME on ordered structures. We are interested
in finding that whether this logic can be a candidate to capture PTIME
on all finite structures. To study the expressive powers of these
logic on all finite structures, we introduce the notions of DATALOG$^{*}$ program and DATALOG$^{r}$
program and their stratified versions (S-DATALOG$^{*}$, S-DATALOG$^{r}$,
respectively). In this paper, we show that SO-HORN$^{r}$ has the same expressive
power as FO(LFP) which is a very powerful logic.  We also study the expressive power of
 SO-EHORN$^{r}$   and prove that SO-EHORN$^{r}$ captures co-NP on all finite structures.
This logic is an extended version of SO-HORN and has been shown that it can capture co-NP on ordered structures.

This paper is organized as follows: In section 2, we set up notations
and terminologies.  In section 3, we study the expressive powers
of SO-HORN$^{*}$, SO-HORN$^{r}$ and SO-HORN$^{*r}$ on all finite
structures. We use DATALOG  and S-DATALOG program as intermediate
logic, and introduce the notions of DATALOG$^{*}$ program,  DATALOG$^{r}$
program and their stratified versions. We show that SO-HORN$^{r}$,
SO-HORN$^{*r}$, FO(LFP) coincide with each other and SO-HORN$^{*}$
is a proper sublogic of SO-HORN$^{r}$. In
section 4, we use the normal forms of $\Sigma_{1}^{1}$ formulas to
show that SO-EHORN$^{r}$ captures co-NP on all finite structures.
We prove this result in two versions. Section 5 is the conclusion
of this paper.

\section{Preliminaries}

We assume the reader has already some familiarity with mathematical
logic and complexity theory. Nonetheless, we give a brief description
of some basic definitions.

A vocabulary is a finite set of relation symbols $P,Q,R,\cdots$ and
constant symbols $\mathbf{c},\mathbf{d},\cdots$. Each relation symbol
is equipped with a natural number $r\geq1$, its arity. We use lower-case
Greek letters $\tau,\sigma,\cdots$ to denote a vocabulary. Given
a vocabulary $\tau=\{P_{1},P_{2},\cdots,P_{n},\mathbf{c}_{1},\mathbf{c}_{2},\cdots,\mathbf{c}_{m}\}$,
a $\tau$-structure $\mathcal{A}$ is a tuple
\[
\left\langle A,P_{1}^{A},P_{2}^{A},\cdots,P_{n}^{A},\mathbf{c}_{1}^{A},\mathbf{c}_{2}^{A},\cdots,\mathbf{c}_{m}^{A}\right\rangle
\]
where $A$ is the domain of $\mathcal{A}$, for each relation symbol
$P_{i}$ of arity $r_{i}\,(1\leq i\leq n)$, $P_{i}^{A}\subseteq A^{r_{i}}$
and for each constant symbol $\mathbf{c}_{j}\,(1\leq j\leq m)$, $\mathbf{c}_{j}^{A}\in A$.
From now on, we use calligraphic letters $\mathcal{A},\mathcal{B},\cdots$
and capital letters $A,B,\dots$ to denote structures and the corresponding
domains, respectively. We use $|A|$ to denote the cardinality of
set $A$ and $\bar{a}$ to denote a sequence $a_{1},a_{2},\cdots,a_{n}$
or a tuple $(a_{1},a_{2},\cdots,a_{n})$ of elements (variables).
A structure $\mathcal{A}$ is finite if its domain $A$ is a finite
set. Without otherwise specified, we restrict ourself to finite structures
in this paper. Let $STRUC(\tau)$ denote the set of all finite $\tau$-structures,
and if $\mathcal{L}$ is a logic, then $\mathcal{L}(\tau)$ denotes
the set of all formulas of $\mathcal{L}$ over the vocabulary $\tau$.
Given a formula $\phi\in\mathcal{L}(\tau)$, define
\[
Mod(\phi)=\{\mathcal{A}\mid\mathcal{A}\textrm{ is a \ensuremath{\tau}-structure and }\mathcal{A}\models\phi\}
\]
Let $\mathcal{L}_{1}$ and $\mathcal{L}_{2}$ be two logic, we use
$\mathcal{L}_{1}\subseteq\mathcal{L}_{2}$ to denote the expressive
power of $\mathcal{L}_{2}$ is no less than that of $\mathcal{L}_{1}$,
i.e. for any vocabulary $\tau$ and any $\phi\in\mathcal{L}_{1}(\tau)$,
there exists $\psi\in\mathcal{L}_{2}(\tau)$ such that $Mod(\phi)=Mod(\psi)$.
$\mathcal{L}_{1}\subset\mathcal{L}_{2}$ denotes that $\mathcal{L}_{2}$
is strictly more expressive than $\mathcal{L}_{1}$, and $\mathcal{L}_{1}\equiv\mathcal{L}_{2}$
denotes that $\mathcal{L}_{1}$ has the same expressive power as $\mathcal{L}_{2}$.

Define $\tau_{<}=\tau\cup\{<,succ,\mathbf{min},\mathbf{max}\}$, where
$\tau$ is a vocabulary. A $\tau_{<}$-structure $\mathcal{A}$ is
ordered if the reduct $\mathcal{A}\mid\{<,succ,\mathbf{min},\mathbf{max}\}$
is an ordering (that is, $<$, $succ$, $\mathbf{min}$ and $\mathbf{max}$
are interpreted by the ordering relation, successor relation, the
least and the last element of the ordering, respectively.) Let $STRUC_{<}(\tau)$
denote the set of all ordered structures over $\tau_{<}$. Given a
formula $\phi\in\mathcal{L}(\tau_{<})$, define
\[
Mod_{<}(\phi)=\{\mathcal{A}\models\phi\mid\mathcal{A}\text{ is an ordered }\tau_{<}\text{-structure}\}
\]

Let $\mathcal{L}$ be a logic and $\mathcal{C}$ be a complexity class,
we say logic $\mathcal{L}$ captures complexity class $\mathcal{C}$
if the following two conditions are satisfied:
\begin{description}
\item [{1)}] The data complexity of $\mathcal{L}$ is in $\mathcal{C}$.

\begin{itemize}
\item For any vocabulary $\tau$ and any closed formula $\phi\in\mathcal{L}(\tau)$,
the membership problem of $Mod(\phi)$ is in $\mathcal{C}$.
\end{itemize}
\item [{2)}] $\mathcal{C}$ is expressible in $\mathcal{L}$.

\begin{itemize}
\item For any vocabulary $\tau$ and any set $K\subseteq STRUC(\tau)$,
if the membership problem of $K$ is in $\mathcal{C}$, then $\exists\phi\in\mathcal{L}(\tau)$
such that $K=Mod(\phi)$.
\end{itemize}
\end{description}
We say $\mathcal{L}$ captures $\mathcal{C}$ on ordered structures
if we only consider ordered structures and it satisfies:
\begin{description}
\item [{1)}] For any vocabulary $\tau$ and any closed formula $\phi\in\mathcal{L}(\tau_{<})$,
the membership problem of $Mod_{<}(\phi)$ is in $\mathcal{C}$.
\item [{2)}] For any vocabulary $\tau$ and any set $K\subseteq STRUC_{<}(\tau)$,
if the membership problem of $K$ is in $\mathcal{C}$, then $\exists\phi\in\mathcal{L}(\tau_{<})$
such that $K=Mod_{<}(\phi)$.
\end{description}
We recall some notions and results of quantified Boolean formulas
(QBF). A QBF formula has the form:
\[
\Phi=Q_{1}x_{1}\cdots Q_{n}x_{n}\varphi
\]
 where $Q_{i}\in\{\forall,\exists\}$ and $\varphi$ is a propositional
CNF formula over Boolean variables. A literal $x_{i}$ or $\neg x_{i}$
is called universal (resp. existential) if $Q_{i}$ is $\forall$
(resp. $\exists$). If every clause in $\varphi$ contains at most
one positive literal (resp. positive existential literal) then $\Phi$
is called a quantified Horn formula (QHORN) (resp. quantified extended
Horn formula (QEHORN)). The evaluation problem for QBF is PSPACE-complete
\cite{kble}, for QHORN is in PTIME, whereas it remains PSPACE-complete
for QEHORN. However, for each fixed $k\geq1$, the evaluation problem
for QEHORN formulas with prefix type $\forall\overline{x}_{1}\exists\overline{y}_{1}\cdots\forall\overline{x}_{k}\exists\overline{y}_{k}$
(here $\overline{x}_{i},\overline{y}_{i}$ are sequences of Boolean
variables) is co-NP complete\cite{fl90,fl93}.

\section{Expressive powers of Horn logic}
\begin{definition}
\label{def:SO-HORN}Second-order Horn logic, denoted by SO-HORN, is
the set of second-order formulas of the form
\[
Q_{1}R_{1}\cdots Q_{m}R_{m}\forall\bar{x}(C_{1}\wedge\cdots\wedge C_{n})
\]
 where $Q_{i}\in\{\forall,\exists\}$, $R_{1},\cdots,R_{m}$ are relation
symbols and $C_{1},\cdots,C_{n}$ are Horn clauses with respect to
$R_{1},\cdots,R_{m}$, more precisely, each $C_{j}$ is an implication
of the form
\[
\alpha_{1}\wedge\cdots\wedge\alpha_{l}\wedge\beta_{1}\wedge\cdots\wedge\beta_{q}\rightarrow H
\]
 where
\begin{description}
\item [{1)}] each $\alpha_{s}$ is an atomic formula $R_{i}\bar{x}$,
\item [{2)}] each $\beta_{t}$ is either an atomic formula $P\bar{y}$
or a negated atomic formula $\neg P\bar{y}$ where $P\notin\{R_{1},\cdots,R_{m}\}$,
\item [{3)}] $H$ is either an atomic formula $R_{k}\overline{z}$ or the
Boolean constant $\bot$ (for false). \end{description}
\begin{itemize}
\item If we replace condition \textbf{2)} by

\begin{description}
\item [{2')}] each $\beta_{t}$ is a first-order formula $\phi(\bar{y})$
containing no $R_{1},\cdots,R_{m}$,
\end{description}

we denote this logic by SO-HORN$^{*}$.

\item If we replace condition \textbf{1)} by

\begin{description}
\item [{1')}] each $\alpha_{s}$ is either an atomic formula $R_{i}\bar{x}$
or $\forall\bar{y}R_{i}\bar{y}\bar{z}$,
\end{description}

we call this logic second-order revised Horn logic, denoted by SO-HORN$^{r}$.

\item If we replace conditions \textbf{1),} \textbf{2)} by \textbf{1')},
\textbf{2')}, respectively, we denote this logic by SO-HORN$^{*r}$.
\end{itemize}
\end{definition}
On ordered structures, Gr\"{a}del\cite{gra91,Gradel} showed that SO-HORN
and SO-HORN$^{*}$ capture PTIME and we\cite{zhao-1} showed that
SO-HORN$^{*r}$ captures PTIME. These four logic coincide with each
other. On all structures, it turns out that
SO-HORN is strictly less expressive than FO(LFP) which is a very powerful
logic. It is not hard to see that the following  holds from the
definition above:
\[
\begin{array}{cccc}
 & \subseteq & \textrm{SO-HORN}^{\mathrm{*}}\\
\textrm{SO-HORN} &  &  & \subseteq\textrm{SO-HORN}^{\mathrm{*r}}\\
 & \subseteq & \textrm{SO-HORN}^{\mathrm{r}}
\end{array}
\]

An important open problem in descriptive complexity is that whether
there is an effective logic captures PTIME on all structures. We are
interesting in finding the expressive powers of these Horn logic on all
structures. At first, we give a lemma that is very useful. Let $\phi$
be a logic formula, we use $\phi[\alpha/\beta]$ to denote the formula
obtained by replacing the formula $\alpha$  by $\beta$ in $\phi$.
\begin{lemma}
\label{lemma:UviFOafterExisSO}Every second-order formula $\forall x\exists R\phi(x,\bar{x})$
is equivalent to a formula of the form $\exists R'\forall x\phi[R\bar{y}_{1}/R'x\bar{y}_{1},\cdots,R\bar{y}_{n}/R'x\bar{y}_{n}](x,\bar{x})$
where $x$ occurs free in $\phi$, $R$ is an $r$-ary relation symbol,
$R'$ is an $r+1$-ary relation symbol and $R\bar{y}_{1},\cdots,R\bar{y}_{n}$
are all the different atomic formulas that occur in $\phi$ and over
the relation symbol $R$.\end{lemma}
\begin{proof}
Given a structure $\mathcal{A}$ and $R'\subseteq A^{r+1}$, write
$$R_{a}=\{(a_{1},\cdots,a_{r})\mid(a,a_{1},\cdots,a_{r})\in R'\}$$
then for any $(a_{1},\cdots,a_{r})\in A^{r}$
\begin{equation}
\begin{array}{ccc}
(a_{1},\cdots,a_{r})\in R_{a} & \textrm{iff} & (a,a_{1},\cdots,a_{r})\in R'\end{array}\label{eq:1}
\end{equation}
Set $\phi'(x,\bar{x})=\phi[R\bar{y}_{1}/R'x\bar{y}_{1},\cdots,R\bar{y}_{n}/R'x\bar{y}_{n}](x,\bar{x})$.
Suppose $\bar{x}=x_{1},\cdots,x_{k}$, for any $\bar{b}=b_{1},\cdots,b_{k}\in A$,
$R'\subseteq A^{r+1}$ and $a\in A$, by \eqref{eq:1} it is easy
to check that
\begin{equation}
\begin{array}{ccc}
(\mathcal{A},R_{a})\models\phi[a,\bar{b}] & \textrm{iff} & (\mathcal{A},R')\models\phi'[a,\bar{b}]\end{array}\label{eq:all_before_exist}
\end{equation}
For any structure $\mathcal{A}$ and $\bar{b}=b_{1},\cdots,b_{k}$,

\[
\begin{array}{lll}
\mathcal{A}\models\forall x\exists R\phi[\bar{b}] & \textrm{iff} & \textrm{for any \ensuremath{a\in A}there exists \ensuremath{\textrm{ }R_{a}\subseteq A^{r}}}\\
 &  & \textrm{such that \ensuremath{(\mathcal{A},R_{a})\models\phi[a,\bar{b}]}}\\
\\
 & \textrm{iff} & (\mathcal{A},R')\models\phi'[a,\bar{b}]\textrm{ for any \ensuremath{a\in A}\,(by \eqref{eq:all_before_exist}, where }\textrm{ }\\
 &  & R'=\underset{a\in A}{\bigcup}\{(a,a_{1},\cdots,a_{r})\mid(a_{1},\cdots,a_{r})\in R_{a}\}\textrm{ )}\\
\\
 & \textrm{iff} & (\mathcal{A},R')\models\forall x\phi'[\bar{b}]\\
\\
 & \textrm{iff} & \mathcal{A}\models\exists R'\forall x\phi'[\bar{b}]
\end{array}
\]
which completes the proof.
\end{proof}
It turns out that SO-HORN collapses to its existential fragment ((SO$\exists$)-HORN,
i.e. the formulas where all the second-order quantifiers are existential\cite{gra91}).
In fact, many second-order formulas equal to the formulas that only
have existential second-order prefix. We extend this result to a more
general form in the following proposition.
\begin{proposition}
Every second-order formula of the form
\begin{equation}
Q_{1}R_{1}\cdots Q_{m}R_{m}\forall\bar{x}\left(\overset{l}{\underset{j=1}{\bigwedge}}(\alpha_{j1}\wedge\cdots\wedge\alpha_{jh_{j}}\wedge\beta_{j1}\wedge\cdots\wedge\beta_{jq_{j}}\rightarrow H_{j})\right)\label{eq:eli_Uni}
\end{equation}
 where $Q_{i}\in\{\forall,\exists\}$ and
\begin{description}
\item [{1)}] each $\alpha_{kt}$ is either an atomic formula $R_{i}\bar{x}$
or $\forall\bar{y}R_{i}\bar{y}\bar{z}$,
\item [{2)}] each $\beta_{ef}$ is a second-order formula that does not
contain $R_{1},\cdots,R_{m}$,
\item [{3)}] $H_{j}$ is either an atomic formula $R_{k}\overline{z}$
or the Boolean constant $\bot$ (for false),
\end{description}
is equivalent to a formula of the form
\begin{equation}
\exists R'_{1}\cdots\exists R'_{n}\forall\bar{x}'\left(\overset{l'}{\underset{j=1}{\bigwedge}}(\alpha'_{j1}\wedge\cdots\wedge\alpha'_{jh'_{j}}\wedge\beta'_{j1}\wedge\cdots\wedge\beta'_{jq'_{j}}\rightarrow H'_{j})\right)\label{eq:eli_Uni_2}
\end{equation}
where
\begin{description}
\item [{1)}] each $\alpha'_{kt}$ is either an atomic formula $R'_{i}\bar{x}$
or $\forall\bar{y}R'_{i}\bar{y}\bar{z}$,
\item [{2)}] each $\beta'_{e'f'}$ is some $\beta_{ef}$ in \eqref{eq:eli_Uni}
that does not contain $R_{1},\cdots,R_{m},R'_{1},\cdots,R'_{n}$,
\item [{3)}] $H'_{j}$ is either an atomic formula $R'_{k}\overline{z}$,
or the Boolean constant $\bot$ (for false).
\end{description}
\end{proposition}
\begin{proof}
This proof is adapted from that of Theorem 5 in \cite{gra91}. It
is suffice to prove for formulas of the form $\forall P\exists R_{1}\cdots\exists R_{n}\forall\bar{x}\phi$.
For an arbitrary second-order prefix we can successively remove the
innermost universal second-order quantifier. The main idea is that for
any structure $\mathcal{A}$
\[
\begin{array}{ccl}
\mathcal{A}\models\forall P\exists R_{1}\cdots\exists R_{n}\forall\bar{x}\phi & \textrm{iff} & (\mathcal{A},P)\models\exists R_{1}\cdots\exists R_{n}\forall\bar{x}\phi\textrm{ for any }P\\
 &  & \textrm{that are false at at most one point.}
\end{array}
\]
Suppose that the arity of $P$ is $k$ and $\bar{y}=y_{1},\cdots,y_{k}$
are new variables that don't occur in $\phi$, we only need to show
that
\[
\begin{array}{ccr}
\mathcal{A}\models\forall P\exists R_{1}\cdots\exists R_{n}\forall\bar{x}\phi & \textrm{iff} & \mathcal{A}\models\forall\bar{y}\exists R{}_{1}\cdots\exists R{}_{n}\forall\bar{x}\phi[P\bar{z}/\bar{z}\neq\bar{y}]\\
 &  & \bigwedge\exists R_{1}\cdots\exists R_{n}\forall\bar{x}\phi[P\bar{z}/\bar{z}=\bar{z}]
\end{array}
\]

The $"\Rightarrow"$ direction is easy. For the $"\Leftarrow"$ direction,
suppose
\[
\mathcal{A}\models\forall\bar{y}\exists R{}_{1}\cdots\exists R{}_{n}\forall\bar{x}\phi[P\bar{z}/\bar{z}\neq\bar{y}]\wedge\exists R_{1}\cdots\exists R_{n}\forall\bar{x}\phi[P\bar{z}/\bar{z}=\bar{z}]
\]
The case that $P=A^{k}$ or $P$ is false at only one tuple is trivial.
Let $P_{\bar{a}}$ be the subset of $A^{k}$ that only the tuple $\bar{a}$
is not in $P$, we use $R_{1}^{P_{\bar{a}}},\cdots,R_{n}^{P_{\bar{a}}}$
to denote the corresponding relations such that $(\mathcal{A},P_{\bar{a}},R_{1}^{P_{\bar{a}}},\cdots,R_{n}^{P_{\bar{a}}})\models\forall\bar{x}\phi$.
For the other cases, consider an arbitrary $P\subset A^{k}$, set
$R_{i}^{P}=\underset{\bar{a}\notin P}{\bigcap}R_{i}^{P_{\bar{a}}}\,(1\leq i\leq n)$.
We proceed to show that
\[
(\mathcal{A},P,R_{1}^{P},\cdots,R_{n}^{P})\models\forall\bar{x}\phi
\]
Conversely, suppose that
\[
(\mathcal{A},P,R_{1}^{P},\cdots,R_{n}^{P})\nvDash\forall\bar{x}\phi
\]
 then there exist $\bar{b}\in A$ and a clause $\alpha_{1}\wedge\cdots\wedge\alpha_{h}\wedge\beta_{1}\wedge\cdots\wedge\beta_{q}\rightarrow H_{j}$
such that
\begin{equation}
(\mathcal{A},P,R_{1}^{P},\cdots,R_{n}^{P})\models\alpha_{1}\wedge\cdots\wedge\alpha_{h}\wedge\beta_{1}\wedge\cdots\wedge\beta_{q}[\bar{b}]\label{eq:eli_Uni_3}
\end{equation}
and

\begin{equation}
(\mathcal{A},P,R_{1}^{P},\cdots,R_{n}^{P})\nvDash H_{j}[\bar{b}]\label{eq:eli_Uni_4}
\end{equation}
Since $P=\underset{\bar{a}\notin P}{\bigcap}P_{\bar{a}}$ and $R_{i}^{P}=\underset{\bar{a}\notin P}{\bigcap}R_{i}^{P_{\bar{a}}}\,(1\leq i\leq n)$,
it follows that for each $\bar{a}\notin P$ and $Q\in\{P,R_{1},\cdots,R_{n}\}$,
\[
(\mathcal{A},P,R_{1}^{P},\cdots,R_{n}^{P})\models Q\bar{c}\Rightarrow(\mathcal{A},P_{\bar{a}},R_{1}^{P_{\bar{a}}},\cdots,R_{n}^{P_{\bar{a}}})\models Q\bar{c}
\]
and
\[
(\mathcal{A},P,R_{1}^{P},\cdots,R_{n}^{P})\models\forall\bar{y}Q\bar{y}\bar{c}'\Rightarrow(\mathcal{A},P_{\bar{a}},R_{1}^{P_{\bar{a}}},\cdots,R_{n}^{P_{\bar{a}}})\models\forall\bar{y}Q\bar{y}\bar{c}'
\]
where $\bar{c}$ and $\bar{c}'$ tuples of elements of $A$. Because
$P,R_{1},\cdots,R_{n}$ don't occur in $\beta_{1},\cdots,\beta_{q}$,
by \eqref{eq:eli_Uni_3} we see that for each $\bar{a}\notin P$
\[
(\mathcal{A},P_{\bar{a}},R_{1}^{P_{\bar{a}}},\cdots,R_{n}^{P_{\bar{a}}})\models\alpha_{1}\wedge\cdots\wedge\alpha_{h}\wedge\beta_{1}\wedge\cdots\wedge\beta_{q}[\bar{b}]
\]
\begin{itemize}
\item If $H_{j}$ is $\bot$ in \eqref{eq:eli_Uni_4}, then $(\mathcal{A},P_{\bar{a}},R_{1}^{P_{\bar{a}}},\cdots,R_{n}^{P_{\bar{a}}})\nvDash\bot$,
contrary to $$(\mathcal{A},P_{\bar{a}},R_{1}^{P_{\bar{a}}},\cdots,R_{n}^{P_{\bar{a}}})\models\forall\bar{x}\phi$$
\item If $H_{j}$ is some $Q\bar{z}$ where $Q\in\{P,R_{1},\cdots,R_{n}\}$,
there must be a $\bar{a}\notin P$ such that $(\mathcal{A},P_{\bar{a}},R_{1}^{P_{\bar{a}}},\cdots,R_{n}^{P_{\bar{a}}})\nvDash H_{j}[\bar{b}]$,
contrary to $$(\mathcal{A},P_{\bar{a}},R_{1}^{P_{\bar{a}}},\cdots,R_{n}^{P_{\bar{a}}})\models\forall\bar{x}\phi$$
\end{itemize}
By repeating use of lemma \ref{lemma:UviFOafterExisSO}, $\forall\bar{y}\exists R{}_{1}\cdots\exists R{}_{n}\forall\bar{x}\phi[P\bar{z}/\bar{z}\neq\bar{y}]\wedge\exists R_{1}\cdots\exists R_{n}\forall\bar{x}\phi[P\bar{z}/\bar{z}=\bar{z}]$
is equivalent to a formula of the required form.\end{proof}
\begin{corollary}
SO-HORN$^{r}$ and SO-HORN$^{*r}$ are both collapse to their existential
fragments.
\end{corollary}
From now on, unless explicitly stated, we will restrict ourselves
to the existential fragments of these logic defined above in the rest
of this paper. To compare the expressive power of the four logic on
all structures, we use DATALOG and S-DATALOG program as intermediate
logic, and introduce the notion of DATALOG$^{r}$ program. For the
detailed definition and semantics of DATALOG program and S-DATALOG
program we refer the reader to \cite{ef}.
\begin{definition}
A DATALOG program $\Pi$ over a vocabulary $\tau$ is a finite set
of rules of the form

\[
\beta\leftarrow\alpha_{1},\cdots,\alpha_{l}
\]
where $l\geq0$ and
\begin{description}
\item [{(1)}] each $\alpha_{i}$ is either an atomic formula or a negated
atomic formula or zero-ary relation symbol,
\item [{(2)}] $\beta$ is either an atomic formula $R\bar{x}$ or a zero-ary
relation symbol $Q$ where $R$ and $Q$ only occur positively in
each rule of $\Pi$.
\end{description}
\end{definition}
$\alpha_{1},\cdots,\alpha_{l}$ constitute the body of the rule, $\beta$
is the head of the rule. Every relation symbols in the head of some
rule of $\Pi$ are intentional, all the other symbols are extensional.
We use $(\tau,\Pi)_{int}$ and $(\tau,\Pi)_{ext}$ to denote the set
of intentional and extensional symbols(constant symbols are extensional symbols), respectively. We allow zero-ary
relation symbols here to be able to compare the expressive power of
DATALOG  with other logic. A zero-ary relation has the Boolean
value TRUE or FALSE. Let $Q$ be a zero-ary relation symbol and $\mathcal{A}$
be a structure, the interpretation of $Q$ in $\mathcal{A}$ is: ``$Q^{A}=\emptyset$''
means ``$Q^{A}=FALSE$'', ``$Q^{A}=\{\emptyset\}$'' means ``$Q^{A}=TRUE$''.

\begin{example}
This is a DATALOG program over the vocabulary $\tau=\{E,R,Q,\mathbf{s},\mathbf{t}\}$
where $(\tau,\Pi)_{int}=\{R,Q\}$, $(\tau,\Pi)_{ext}=\{E,\mathbf{s},\mathbf{t}\}$
and $\mathbf{s},\mathbf{t}$ are constant symbols.
\[
\begin{array}{cccl}
\Pi: & Rxy & \leftarrow & Exy\\
 & Rxy & \leftarrow & Exz,\, Rzy\\
 & Q & \leftarrow & R\mathbf{st}
\end{array}
\]

\end{example}
Instead of giving a detailed definition of the semantics of DATALOG
programs, we present it briefly using an example. Given a DATALOG
program $\Pi$ and a $(\tau,\Pi)_{ext}$-structure, we can apply all
the rules of $\Pi$ simultaneously to generate consecutive stages
of the intentional symbols. Consider the above example, let $R_{(0)}=\emptyset$,
$Q_{(0)}=\emptyset$ and $\mathcal{A}$ be a $(\tau,\Pi)_{ext}$-structure.
Suppose $R_{(i)}$ and $Q_{(i)}$ are known, $R_{(i+1)}$ and $Q_{(i+1)}$
can be evaluated by

\[
\begin{array}{lcl}
R_{(i+1)}xy & \leftarrow & Exy\\
R_{(i+1)}xy & \leftarrow & Exz,\, R_{(i)}zy\\
Q_{(i+1)} & \leftarrow & R_{(i)}\mathbf{st}
\end{array}
\]
where $Q_{(i+1)}=\{\emptyset\}$ iff $\mathcal{A}\models R_{(i)}\mathbf{st}$.
Because $R$ occurs only positively in the body of each rules, we
can reach a fixed-point in polynomial many steps with respect to $|A|$.
Let $R_{(\infty)}=\underset{n\geq0}{\bigcup}R_{(n)}$ denote the fixed-point.
Then $\mathcal{A}$ gives rise to a $\tau$-structure
\[
\mathcal{A}[\Pi]=(\mathcal{A},R_{(\infty)},Q_{(\infty)})
\]

A DATALOG formula has the form $(\Pi,P)\bar{t}$ where $P$ is an
$r$-ary intentional relation symbol and $\bar{t}=t_{1},\cdots,t_{r}$
are variables that don't occur in $\Pi$. Given a $(\tau,\Pi)_{ext}$-structure
$\mathcal{A}$ and $\bar{a}=a_{1},\cdots,a_{r}\in A$
\[
\begin{array}{ccc}
\mathcal{A}\models(\Pi,P)[\bar{a}] & \textrm{iff} & (a_{1},\cdots,a_{r})\in P^{\mathcal{A}[\Pi]}\end{array}
\]

If $P$ is a zero-ary relation symbol, then
\[
\begin{array}{ccc}
\mathcal{A}\models(\Pi,P) & \textrm{iff} & P^{\mathcal{A}[\Pi]}\end{array}\textrm{is TRUE}
\]

In the above example, given $a,b\in A$,

\[
\begin{array}{lcccc}
\mathcal{A}\models(\Pi,R)[a,b] & \textrm{iff} & (a,b)\in R^{\mathcal{A}[\Pi]} & \textrm{iff} & \textrm{there is a directed path from a to b,}\\
\\
\mathcal{A}\models(\Pi,Q) & \textrm{iff} & Q^{\mathcal{A}[\Pi]}\textrm{ is TRUE} & \textrm{iff} & \textrm{there is a directed path from \ensuremath{\mathbf{s}} to  \ensuremath{\mathbf{t}}.}
\end{array}
\]

\begin{definition}
We define three extensions of DATALOG programs:
\begin{itemize}
\item If we allow first-order formulas only over extensional symbols in
the body of rules of a DATALOG program, we denote this logic program
by DATALOG$^{*}$ program.
\item If we allow formulas of the form $\forall\bar{y}R\bar{y}\bar{z}$
where $R$ is an intentional relation symbol in the body of rules
of a DATALOG program, we denote this logic program by DATALOG$^{r}$
program.
\item If we allow both first-order formulas only over extensional symbols
and formulas of the form $\forall\bar{y}R\bar{y}\bar{z}$ where $R$
is an intentional relation symbol in the body of rules of a DATALOG
program, we denote this logic program by DATALOG$^{*r}$ program.
\end{itemize}
\end{definition}
\begin{example}
This is a DATALOG$^{*r}$ program over $\tau=\{R,Q,P_{1},\cdots,P_{n}\}$
where $(\tau,\Pi)_{int}=\{R,Q\}$, $(\tau,\Pi)_{ext}=\{P_{1},\cdots,P_{n}\}$
and $\phi(x,y)$ is a first-order formula over $(\tau,\Pi)_{ext}$.
\[
\begin{array}{cccl}
\Pi: & Rxy & \leftarrow & \phi(x,y)\\
 & Rxy & \leftarrow & \phi(x,z),\, Rzy\\
 & Qx & \leftarrow & \forall yRxy
\end{array}
\]
Relation $R$ denotes the transitive closure of the graph defined
by $\phi(x,y)$. $Q$ denotes the set of vertices that can reach every
vertex in this graph.
\end{example}
For a DATALOG ( DATALOG$^{*}$, DATALOG$^{r}$, DATALOG$^{*r}$, respectively
) program $\Pi$, if there is a zero-ary relation symbol occurring
in the body of some rules, e.g. $Q$, we replace every occurrences
of $Q$ by $Q'x$, where $Q'$ is a new relation symbol and $x$ is
a new variable that do not occur in $\Pi$, and then add the rule
$Q\leftarrow Q'x$ in $\Pi$. Let $\Pi'$ denote the resulting DATALOG
( DATALOG$^{*}$, DATALOG$^{r}$, DATALOG$^{*r}$, respectively )
program. It is easily check that the fixed-points of all common intentional
relations of $\Pi$ and $\Pi'$ coincide on all structures. So for
technical reason, unless stated explicitly, we restrict that the zero-ary
relation symbols only occur in the head of a rule.

It turns out that many first-order definable properties can not be
defined by DATALOG formulas. DATALOG$^{*}$ and DATALOG$^{*r}$ programs
which are equipped with first-order formulas are more expressive than
DATALOG programs. In the following we show that DATALOG$^{r}$ formulas
is enough to express all first-order definable properties. We say
a DATALOG ( DATALOG$^{*}$, DATALOG$^{r}$, DATALOG$^{*r}$, respectively)
formula $(\Pi,P)\bar{t}$ is bounded if there exists a fixed number
$k\geq0$ such that $P_{(k)}=P_{(\infty)}$ for all structures.
\begin{proposition}
\label{pro:FOeqToBndD^r}Every first-order formula is equivalent to
a bounded DATALOG$^{r}$ formula. \end{proposition}
\begin{proof}
We assign every first-order formula $\phi(\bar{x})$ an equivalent
bounded DATALOG$^{r}$ formula $(\Pi_{\phi},P_{\phi})\bar{t}$. We
prove this result by induction on the quantifier rank of $\phi(\bar{x})$.

With no loss of generality, suppose $\phi(\bar{x})$ is in DNF form
\[
Q_{1}x_{1}\cdots Q_{m}x_{m}(C_{1}\vee\cdots\vee C_{n})
\]
where $Q_{i}\in\{\forall,\exists\}$ and each $C_{j}$ is a conjunction
of atomic or negated atomic formulas. Assume $\bar{v}=v_{1},\cdots,v_{r}$
are all the variables in $C_{1}\vee\cdots\vee C_{n}$,
\begin{itemize}
\item for each $C_{j}=\alpha_{j1}\wedge\alpha_{j2}\wedge\cdots\wedge\alpha_{jk_{j}}$,
set
\[
\Pi_{C_{j}}=\{P_{C_{j}}\bar{v}\leftarrow\alpha_{j1},\alpha_{j2},\cdots,\alpha_{jk_{j}}\}
\]

\item set
\[
\Pi_{C_{1}\vee\cdots\vee C_{n}}=\overset{n}{\underset{j=1}{\bigcup}}\Pi_{C_{j}}\cup\overset{n}{\underset{j=1}{\bigcup}}\{P_{C_{1}\vee\cdots\vee C_{n}}\bar{v}\leftarrow P_{C_{j}}\bar{v}\}
\]

\end{itemize}

It's easily seen that $C_{1}\vee\cdots\vee C_{n}(\bar{v})$ and $(\Pi_{C_{1}\vee\cdots\vee C_{n}},P_{C_{1}\vee\cdots\vee C_{n}})\bar{t}$
are equivalent. Suppose $\psi(\bar{x},x)$ and $(\Pi_{\psi},P_{\psi})\bar{t}$
are equivalent and $P_{\phi}$ is a new relation symbol not in $\Pi_{\psi}$,
\begin{itemize}
\item if $\phi(\bar{x})=\forall x\psi(\bar{x},x)$, set
\[
\Pi_{\phi}=\Pi_{\psi}\cup\{P_{\phi}\bar{x}\leftarrow\forall xP_{\psi}\bar{x}x\}
\]

\item if $\phi(\bar{x})=\exists x\psi(\bar{x},x)$, set
\[
\Pi_{\phi}=\Pi_{\psi}\cup\{P_{\phi}\bar{x}\leftarrow P_{\psi}\bar{x}x\}
\]

\end{itemize}

It is easy to check that $\phi(\bar{x})$ and $(\Pi_{\phi},P_{\phi})\bar{t}$
are equivalent. To prove $(\Pi_{\phi},P_{\phi})\bar{t}$ is bounded,
observe that each $P_{C_{j}}$ can reach its fixed-point in one step
and $P_{C_{1}\vee\cdots\vee C_{n}}$ can reach its fixed-point in
two steps. Suppose the quantifier rank of $\phi(\bar{x})$ is $m$,
$P_{\phi}$ can reach its fixed-point in $m+2$ steps.

\end{proof}
\begin{corollary}
\label{cor:D^*D^rD^*r}DATALOG$^{*}\subseteq$ DATALOG$^{r}\equiv$
DATALOG$^{*r}$ \end{corollary}
\begin{proof}
Suppose $(\Pi,P)\bar{t}$ is a DATALOG$^{*}$( DATALOG$^{*r}$ ) formula
over vocabulary $\tau$ and $\phi_{1}(\bar{x}_{1}),\cdots,\phi_{n}(\bar{x}_{n})$
are the no atomic (or negated atomic) first-order formulas that occur
in $\Pi$ and over $(\tau,\Pi)_{ext}$. Let $(\Pi_{\phi_{i}},P_{\phi_{i}})\bar{t}\,(1\leq i\leq n)$
be the corresponding DATALOG$^{r}$ formulas that equivalent to $\phi_{i}(\bar{x}_{i})\,(1\leq i\leq n)$,
respectively. Set
\[
\Pi'=\Pi[\phi_{1}(\bar{x}_{1})/P_{\phi_{1}}\bar{x}_{1},\cdots,\phi_{n}(\bar{x}_{n})/P_{\phi_{n}}\bar{x}_{n}]\cup\overset{n}{\underset{i=1}{\bigcup}}\Pi_{\phi_{i}}
\]
where $\Pi[\phi_{1}(\bar{x}_{1})/P_{\phi_{1}}\bar{x}_{1},\cdots,\phi_{n}(\bar{x}_{n})/P_{\phi_{n}}\bar{x}_{n}]$
denotes the logic program obtained by replacing each $\phi_{i}(\bar{x}_{i})$
by $P_{\phi_{i}}\bar{x}_{i}\,(1\leq i\leq n)$ in $\Pi$. The DATALOG$^{r}$
formula $(\Pi',P)\bar{t}$ is equivalent to $(\Pi,P)\bar{t}$.
\end{proof}
A DATALOG program $\Pi$ is positive if no atomic occurs negatively
in any rule of $\Pi$. M. Ajtai and Y. Gurevich\cite{ajtai} showed
that a positive DATALOG formula is bounded if and only if it is equivalent
to a (existential positive) first-order formula. For a DATALOG formula,
if it is bounded then it equals to a first-order formula, but the
converse is false. They find a DATALOG formula which is equivalent
to a first-order formula but not bounded. We show that this statement
is true for DATALOG$^{r}$ formulas.
\begin{theorem}
\label{thm:BndnessOfD^r}For any DATALOG$^{r}$ formula $(\Pi,P)\bar{t}$
the following are equivalent:
\begin{description}
\item [{(\mbox{i})}] $(\Pi,P)\bar{t}$ is equivalent to a first-order formula.
\item [{(\mbox{ii})}] $(\Pi,P)\bar{t}$ is equivalent to a bounded DATALOG$^{r}$
formula.
\end{description}
\end{theorem}
\begin{proof}
(i) $\Rightarrow$ (ii) is by Proposition \ref{pro:FOeqToBndD^r}.
The proof for (ii) $\Rightarrow$ (i) is similar to that of Proposition
9.3.1 in \cite{ef}, we give only the main ideas. Observing that given
a DATALOG$^{r}$ formula $(\Pi,P)\bar{t}$ and a fixed $k$, the $k$-th
stage $P_{(k)}$ in the evaluation is expressible by a first-order
formula. If $(\Pi,P)\bar{t}$ is bounded then there exists an $l\geq0$
such that $P_{(l)}=P_{(\infty)}$ on all structures.\end{proof}
\begin{remark}
If a DATALOG formula is equivalent to a bounded DATALOG formula, then
the DATALOG formula itself is bounded. But if a DATALOG$^{r}$ formula
is equivalent to a bounded DATALOG$^{r}$ formula, itself is not necessary
bounded. We give a counterexample. Consider an atomic formula $Px$,
let $<$ and $succ$ be two 2-ary relation symbols and $\mathbf{min}$
be a constant symbol. Construct two first-order formulas $\psi_{1}(<,\mathbf{min})$
and $\psi_{2}(succ)$, where $\psi_{1}(<,\mathbf{min})$ says that
$<$ is a linear ordering relation and $\mathbf{min}$ is the least
element, $\psi_{2}(succ)$ says that $succ$ is a successor relation.
By Proposition \ref{pro:FOeqToBndD^r}, $\psi_{1}$ is equivalent
to a DATALOG$^{r}$ formula $(\Pi_{1},P_{<})$, $\neg\psi_{1}$ is
equivalent to a DATALOG$^{r}$ formula $(\Pi'_{1},P_{\nless})$, $\psi_{2}$
is equivalent to a DATALOG$^{r}$ formula $(\Pi_{2},P_{succ})$ and
$\neg\psi_{2}$ is equivalent to a DATALOG$^{r}$ formula $(\Pi'_{2},P_{\neg succ})$.
Set
\[
\begin{array}{lrcl}
\Pi_{3}: & Q\mathbf{min} & \leftarrow & P_{<},P_{succ}\\
 & Qy & \leftarrow & Qx,succ(x,y),P_{<},P_{succ}\\
 & P'x & \leftarrow & Qx,Px,P_{<},P_{succ}\\
 & P'x & \leftarrow & Px,P_{<},P_{\neg succ}\\
 & P'x & \leftarrow & Px,P_{\nless}
\end{array}
\]
and $\Pi=\Pi_{1}\cup\Pi_{2}\cup\Pi'_{1}\cup\Pi'_{2}\cup\Pi_{3}$.
The DATALOG$^{r}$ program $\Pi_{3}$ says that if $<$ is a linear
ordering relation, $\mathbf{min}$ is the least element and $succ$
is a successor relation, then it checks all elements one by one that
whether they have the property $P$. Otherwise, just let $P'=P$.
The DATALOG$^{r}$ formula $(\Pi,P')t$ is equivalent to $Px$ which
is a first-order formula, but for a $\{P,<,succ,\mathbf{min}\}$-structure
$\mathcal{A}$ where $<$ interpreted as a linear ordering relation,
$\mathbf{min}$ interpreted as the least element and $succ$ interpreted
as a successor relation, $P'$ reach its fixed-point in more than
$|A|$ steps.
\end{remark}
Corollary \ref{cor:D^*D^rD^*r} shows that DATALOG$^{*}$ programs
is a sublogic of DATALOG$^{r}$ programs. In fact, it is proper.
Before proving this result, we review stratified DATALOG program
( by short: S-DATALOG program ) which is more powerful than DATALOG
program.
\begin{definition}
\label{def:S-DATALOG}A stratified DATALOG program $\Sigma$, denoted
by S-DATALOG program $\Sigma$, is a sequence $\Pi_{0},\Pi_{1},\cdots,\Pi_{n}$
of DATALOG programs over vocabularies $\tau_{0},\tau_{1},\cdots,\tau_{n}$,
respectively, such that $(\tau_{m+1},\Pi_{m+1})_{ext}=\tau_{m}\,(0\leq m<n)$.
\end{definition}
Given a S-DATALOG program $\Sigma=(\Pi_{0},\Pi_{1},\cdots,\Pi_{n})$
and a $(\tau_{0},\Pi_{0})_{ext}$-structure $\mathcal{A}$, we set
\[
\mathcal{A}[\Sigma]=(\ldots(\mathcal{A}[\Pi_{0}])[\Pi_{1}]\ldots)[\Pi_{n}]
\]
A S-DATALOG formula has the form $(\Sigma,P)\bar{t}$ where $P$ is
an $r$-ary intentional relation symbol of one of the constitutes
$\Pi_{i}$ and $\bar{t}=t_{1},\cdots,t_{r}$ are variables that do
not occur in $\Sigma$.
\begin{example}
The following is an S-DATALOG program $\Sigma=(\Pi_{0},\Pi_{1})$
over the vocabulary $\tau_{0}=\{E,R\},\tau_{1}=\{E,R,P\}$ where $(\tau_{0},\Pi_{0})_{int}=\{R\}$,
$(\tau_{0},\Pi_{0})_{ext}=\{E\}$, $(\tau_{1},\Pi_{1})_{int}=\{P\}$
and $(\tau_{1},\Pi_{1})_{ext}=\{E,R\}$.
\[
\begin{array}{crcl}
\Pi_{0}: & Rxy & \leftarrow & Exy\\
 & Rxy & \leftarrow & Exz,Rzy
\end{array}\begin{array}{cccc}
\Pi_{1}: & Pxy & \leftarrow & \neg Rxy\\
\\
\end{array}
\]

For any $(\tau_{0},\Pi_{0})_{ext}$-structure $\mathcal{A}$ and $a,b\in A$,
\[
\begin{array}{lcccc}
\mathfrak{A}\models(\Sigma,R)[a,b] & \textrm{iff} & (a,b)\in R^{\mathcal{A}[\Sigma]} & \textrm{iff} & \textrm{there is a directed path from a to b,}\\
\\
\mathfrak{A}\models(\Sigma,P)[a,b] & \textrm{iff} & (a,b)\in P^{\mathcal{A}[\Sigma]} & \textrm{iff} & \textrm{there is no directed path from a to b.}
\end{array}
\]
\end{example}
\begin{definition}
In Definition \ref{def:S-DATALOG},
\begin{itemize}
\item if $\Sigma$ is a sequence $\Pi_{0},\Pi_{1},\cdots,\Pi_{n}$ of DATALOG$^{*}$
programs, we denote this logic program by S-DATALOG$^{*}$ program,
\item if $\Sigma$ is a sequence $\Pi_{0},\Pi_{1},\cdots,\Pi_{n}$ of DATALOG$^{r}$
programs, we denote this logic program by S-DATALOG$^{r}$ program.
\end{itemize}
\end{definition}
\begin{proposition}
\label{pro:S-DALG=S-DALG^*}S-DATALOG $\equiv$ S-DATALOG$^{*}$\end{proposition}
\begin{proof}
The main idea of the proof is that every first-order formula $\phi(\bar{x})$
is equivalent to a S-DATALOG formula $(\Sigma_{\phi},P_{\phi})\bar{t}$.
So we can replace $\phi(\bar{x})$ in S-DATALOG$^{*}$ program $\Sigma$
by $P_{\phi}\bar{x}$ and add $\Sigma_{\phi}$ to $\Sigma$. The details
are left to the reader.\end{proof}
\begin{theorem}
\label{thm:S-DALG^r=FO(LFP)}DATALOG$^{r}\equiv$ S-DATALOG$^{r}\equiv$
FO(LFP)\end{theorem}
\begin{proof}
We prove by showing that DATALOG$^{r}\subseteq$ S-DATALOG$^{r}\subseteq$
FO(LFP) $\subseteq$ DATALOG$^{r}$. The containment DATALOG$^{r}\subseteq$
S-DATALOG$^{r}$ is trivial. To prove S-DATALOG$^{r}\subseteq$ FO(LFP),
we first consider the case when $(\Pi,P)\bar{t}$ is a DATALOG$^{r}$
formula. For each intentional relation symbol $Q$ in $\Pi$, we construct
a first-order formula
\[
\begin{array}{r}
\phi_{Q}(\bar{x}_{Q})=\bigvee\{\exists\bar{y}(\gamma_{1}\wedge\cdots\wedge\gamma_{l})\mid Q\bar{x}_{Q}\leftarrow\gamma_{1},\cdots,\gamma_{l}\in\Pi\textrm{ and }\bar{y}\\
\textrm{are the free variables in \ensuremath{\gamma_{1},\cdots,\gamma_{l}}except \ensuremath{\bar{x}_{Q}}}\}
\end{array}
\]
By the semantics of DATALOG$^{r}$, $(\Pi,P)\bar{t}$ is equivalent
to the FO(S-LFP) formula (for details of FO(S-LFP) we refer the reader
to \cite{ef})
\[
[\textrm{S-LFP}_{\bar{x}_{P},\, P,\,\bar{x}_{Q_{1}},\, Q_{1},\cdots,\bar{x}_{Q_{n}},\, Q_{n}}\phi_{P},\phi_{Q_{1}},\cdots,\phi_{Q_{n}}]\bar{t}
\]
which is equivalent to a FO(LFP) formula. For an S-DATALOG$^{r}$
formula $(\Sigma,P)\bar{t}$ which over the program $\Sigma=(\Pi_{0},\Pi_{1},\cdots,\Pi_{n})$,
we prove by induction on $n$. For each formula $(\Pi_{0},P)\bar{t}$
we find an equivalent FO(LFP) formula as above. Suppose every formula
$(\Pi_{i},P)\bar{t}\,(0\leq i\leq j)$, where $P$ is an intentional
relation symbol of $\Pi_{i}$, has an equivalent FO(LFP) formula,
we replace $P\bar{x}$ by the equivalent FO(LFP) formula in $\Pi_{j+1}$
and deal with $\Pi_{j+1}$ as a DATALOG$^{r}$ program.

To prove FO(LFP) $\subseteq$ DATALOG$^{r}$, we use the normal form
of FO(LFP) formulas. It has been shown that every FO(LFP) formula
$\phi(\bar{x})$ is equivalent to a formula of the form (Theorem 9.4.2,
\cite{ef})
\[
\exists u[\textrm{LFP}_{\bar{z},Z}\psi(\bar{x},\bar{z})]\tilde{u}
\]
where $\psi$ is a first-order formula. Without loss of generality,
we assume that $\psi$ is in DNF form
\[
Q_{1}x_{1}\cdots Q_{m}x_{m}(C_{1}\vee\cdots\vee C_{n})
\]
where $Q_{i}\in\{\forall,\exists\}$ and each $C_{j}$ is a conjunction
of atomic or negated atomic formulas. By Proposition \ref{pro:FOeqToBndD^r}
we know that there is a DATALOG$^{r}$ formula $(\Pi,P)\bar{x}\bar{z}$
such that $\psi(\bar{x},\bar{z})$ and $(\Pi,P)\bar{x}\bar{z}$ are
equivalent. Set
\[
\Pi'=\Pi[Z\bar{z}/Z\bar{x}\bar{z}]\cup\{Z\bar{x}\bar{z}\leftarrow P\bar{x}\bar{z},Q\bar{x}\leftarrow Z\bar{x}\tilde{u}\}
\]
where $\Pi[Z\bar{z}/Z\bar{x}\bar{z}]$ denotes the logic program obtained
by replacing each $Z\bar{z}$ by $Z\bar{x}\bar{z}$ in $\Pi$. By
the definition of FO(LFP) we know that $Z$ occurs only positively
in each $C_{j}$, so $\Pi'$ is a DATALOG$^{r}$ program and $(\Pi',Q)\bar{t}$
is equivalent to $\phi(\bar{x})$.
\end{proof}
P. Kolaitis\cite{pk} showed that S-DATALOG is equivalent to EFP,
H. Ebbinghaus and J. Flum\cite{ef} showed that S-DATALOG is equivalent
to BFP. Both EFP and BFP are proper subsets of FO(LFP). Combining
Proposition \ref{pro:S-DALG=S-DALG^*} and Theorem \ref{thm:S-DALG^r=FO(LFP)}
we conclude the following corollary.
\begin{corollary}
DATALOG$^{*}\subset$ DATALOG$^{r}$
\end{corollary}
It is natural to try to relate Horn logic to Logic programs. First
we prove a lemma that shows a property of SO-HORN formulas. We say
a formula $\phi$ is closed under substructures if for any structure
$\mathcal{A}$ and any $\mathcal{B}\subseteq\mathcal{A}$, $\mathcal{A}\models\phi$
implies $\mathcal{B}\models\phi$. A formula $\phi$ is closed under
extensions if for any structure $\mathcal{A}$ and any $\mathcal{A}\subseteq\mathcal{B}$,
$\mathcal{A}\models\phi$ implies $\mathcal{B}\models\phi$. The following
lemma shows that adding a second-order quantifier does not change
these preservation  properties of a formula.
\begin{lemma}
\label{lem:clo_under_subsruc}If $\phi$ is a second-order formula
closed under substructures (extensions), then both $\exists R\phi$
and $\forall R\phi$ are closed under substructures (extensions).\end{lemma}
\begin{proof}
Given a structure $\mathcal{A}$ and an $r$-ary relation $R$ over
$\mathcal{A}$, if $\mathcal{B}\subseteq\mathcal{A}$, then we use
$R\cap B^{r}$ to denote the reduct of $R$ to $\mathcal{B}$. Suppose
that $\phi$ is closed under substructures, then for any structure
$\mathcal{A}$ and $\mathcal{B}\subseteq\mathcal{A}$,
\[
\begin{array}{ccl}
\mathcal{A}\models\exists R\phi & \Rightarrow & \textrm{there exists a }R\subseteq A^{r}\\
 &  & \textrm{such that }(\mathcal{A},R)\models\phi\\
 & \Rightarrow & (\mathcal{B},R\cap B^{r})\models\phi\\
 &  & (\phi\textrm{ is closed under substructures})\\
 & \Rightarrow & \mathcal{B}\models\exists R\phi.
\end{array}
\]

\[
\begin{array}{ccl}
\mathcal{A}\models\forall R\phi & \Rightarrow & (\mathcal{A},R)\models\phi\textrm{ for any }R\subseteq A^{r}\\
 & \Rightarrow & (\mathcal{B},R\cap B^{r})\models\phi\textrm{ for any }R\subseteq A^{r}\\
 &  & (\phi\textrm{ is closed under substructures})\\
 & \Rightarrow & (\mathcal{B},R')\models\phi\textrm{ for any }R'\subseteq B^{r}\\
 & \Rightarrow & \mathcal{B}\models\forall R\phi.
\end{array}
\]
 If $\phi$ is closed under extensions, then for any structure $\mathcal{A}$
and $\mathcal{A}\subseteq\mathcal{B}$
\[
\begin{array}{ccl}
\mathcal{A}\models\exists R\phi & \Rightarrow & \textrm{there exists a }R\subseteq A^{r}\\
 &  & \textrm{such that }(A,R)\models\phi\\
 & \Rightarrow & (\mathcal{B},R)\models\phi\\
 &  & (\phi\textrm{ is closed under extensions})\\
 & \Rightarrow & \mathcal{B}\models\exists R\phi.
\end{array}
\]

\[
\begin{array}{ccl}
\mathcal{A}\models\forall R\phi & \Rightarrow & (\mathcal{A},R)\models\phi\textrm{ for any }R\subseteq A^{r}\\
 & \Rightarrow & (\mathcal{A},R'\cap A^{r})\models\phi\textrm{ for any }R'\subseteq B^{r}\\
 & \Rightarrow & (\mathcal{B},R')\models\phi\textrm{ for any }R'\subseteq B^{r}\\
 &  & (\phi\textrm{ is closed under extensions})\\
 & \Rightarrow & \mathcal{B}\models\forall R\phi.
\end{array}
\]
\end{proof}
\begin{corollary}
SO-HORN is closed under substructures.\end{corollary}
\begin{proof}
It's easily seen that for a SO-HORN formula $\Phi=Q_{1}R_{1}\cdots Q_{n}R_{n}\forall\bar{x}\phi$
where $Q_{i}\in\{\forall,\exists\}\,(1\leq i\leq n)$, the first-order
part $\forall\bar{x}\phi$ is closed under substructures. By Lemma
\ref{lem:clo_under_subsruc}, $\Phi$ is closed under substructures.
\end{proof}
For any formula $\phi$, $\phi$ is closed under substructures if
and only if $\neg\phi$ is closed under extensions. M. Ajtai and Y.
Gurevich\cite{ajtai} mentioned that DATALOG formulas are closed under
extensions. The following proposition implies that, in a sense, DATALOG
is equivalent to the negation of SO-HORN.
\begin{proposition}
\label{pro:HORN=NegDATALOG}For every SO-HORN ( SO-HORN$^{*}$,
SO-HORN$^{r}$, SO-HORN$^{*r}$, respectively ) formula $\phi(\bar{x})$
over vocabulary $\sigma$ we can find a DATALOG ( DATALOG$^{*}$,
DATALOG$^{r}$, DATALOG$^{*r}$, respectively ) formula $(\Pi,P)\bar{t}$
over vocabulary $\tau$ such that $\sigma=(\tau,\Pi)_{ext}$ and the
following are satisfied:
\begin{itemize}
\item If $\phi$ is a sentence, then $P$ is a zero-ary intentional relation
symbol such that for any $(\tau,\Pi)_{ext}$-structure $\mathcal{A}$
\[
\begin{array}{cccc}
\mathcal{A}\models\phi & \textrm{iff} & \mathcal{A}\nvDash(\Pi,P) & \textrm{i.e. }P^{\mathcal{A}[\Pi]}\textrm{ is FALSE}\end{array}
\]

\item If $\phi(\bar{x})$ is a formula with free variables $\bar{x}=x_{1},\cdots,x_{r}$,
then $P$ is an $r$-ary intentional relation symbol such that for
any $(\tau,\Pi)_{ext}$-structure $\mathcal{A}$ and any $\bar{a}=a_{1},\cdots,a_{r}\in A$
\[
\begin{array}{cccc}
\mathcal{A}\models\phi[\bar{a}] & \textrm{iff} & \mathcal{A}\nvDash(\Pi,P)[\bar{a}] & \textrm{i.e. }(a_{1},\cdots,a_{r})\notin P^{\mathcal{A}[\Pi]}\end{array}
\]

\end{itemize}

The converse of this statement is also true.

\end{proposition}
\begin{proof}
We give the proof only for the case of DATALOG and SO-HORN, the others
follow by the same method. For every SO-HORN formula $\phi(\bar{x})$
we construct a DATALOG formula $(\Pi,P)\bar{t}$ which is equivalent
to $\neg\phi(\bar{x})$. Suppose $\phi(\bar{x})$ has the form
\[
\exists R_{1}\cdots\exists R_{n}\forall x_{1}\cdots\forall x_{m}\left(\overset{l}{\underset{i=1}{\bigwedge}}(\alpha_{i1}\wedge\cdots\wedge\alpha_{is_{i}}\rightarrow H_{i})\right)(\bar{x})
\]
where $\bar{x}=x_{1},\cdots,x_{r}$ are the free variables in $\phi$.
Suppose $\bar{v}=v_{1},\cdots,v_{r}$ are the variables that don't
occur in $\alpha$, we use
\[
\alpha'=\alpha[R{}_{1}\bar{u}_{1}/R{}_{1}'\bar{u}_{1}\bar{v},\cdots,R{}_{n}\bar{u}_{n}/R{}_{n}'\bar{u}_{n}\bar{v}]
\]
to denote the formula obtained by replacing each $R_{i}\bar{u}_{i}$
by $R_{i}'\bar{u}_{i}\bar{v}\,(1\leq i\leq n)$, respectively. We
construct $\Pi$ as following:
\begin{itemize}
\item for each clause $\alpha_{i1}\wedge\cdots\wedge\alpha_{is_{i}}\rightarrow H_{i}$
in $\phi$,

\begin{itemize}
\item if $H_{i}$ is an atomic formula $R_{j}\bar{u}$, we add the following
rule in $\Pi$,
\[
R_{j}'\bar{u}\bar{v}\leftarrow\bar{x}=\bar{v},\alpha'_{i1},\cdots,\alpha'_{is_{i}}
\]

\item if $H_{i}$ is the Boolean constant $\bot$, we add the following
rule in $\Pi$,
\[
P\bar{v}\leftarrow\bar{x}=\bar{v},\alpha'_{i1},\cdots,\alpha'_{is_{i}}
\]

\end{itemize}
\item if no $H_{i}$ is the Boolean constant $\bot$ in $\phi$, we add
the following rule in $\Pi$,
\[
P\bar{v}\leftarrow P\bar{v}
\]

\end{itemize}

It's easily seen that if $\phi$ is a sentence then $P$ is a zero-ary
intentional relation symbol in $\Pi$. From the semantics of DATALOG
program and SO-HORN logic we know that for any structure $\mathcal{A}$
and $\bar{a}=a_{1},\cdots,a_{r}\in A$, $\mathcal{A}\models\phi[\bar{a}]$
iff $\mathcal{A}\nvDash(\Pi,P)[\bar{a}]$.

For the converse direction, given a DATALOG formula $(\Pi,P)\bar{t}$,
we construct a SO-HORN formula $\phi(\bar{x})$ such that $(\Pi,P)\bar{t}$
and $\neg\phi(\bar{x})$ are equivalent. Remember that we assume no
zero-ary relation symbol in the body of any rule of $\Pi$, so we
can ignore the rules $Q\leftarrow\alpha_{1},\cdots,\alpha_{l}$ in
$\Pi$ where $Q$ is a zero-ary relation symbol and $Q\neq P$ in
the following construction with no result affected. Let $R_{1},\cdots,R_{n}$
be the all no zero-ary intentional relation symbols and $z_{1},\cdots,z_{m}$
be the all free variables in $\Pi$. The prefix of $\phi(\bar{x})$
is $\exists R_{1}\cdots\exists R_{n}\forall z_{1}\cdots\forall z_{m}$,
the matrix of $\phi(\bar{x})$ is obtained by
\begin{itemize}
\item if $P$ is a zero-ary relation symbol,

\begin{itemize}
\item if $P\leftarrow\alpha_{1},\cdots,\alpha_{l}$ is a rule of $\Pi$,
then we add the clause $\alpha_{1}\wedge\cdots\wedge\alpha_{l}\rightarrow\bot$
as a conjunct,
\item if $R_{i}\bar{u}\leftarrow\alpha_{1},\cdots,\alpha_{l}$ is a rule
of $\Pi$, then we add the clause $\alpha_{1}\wedge\cdots\wedge\alpha_{l}\rightarrow R_{i}\bar{u}$
as a conjunct,
\end{itemize}
\item if $P$ is an $r$-ary $(r>0)$ relation symbol,

\begin{itemize}
\item if $P\bar{u}\leftarrow\alpha_{1},\cdots,\alpha_{l}$ is a rule of
$\Pi$, then let $\bar{x}=x_{1},\cdots,x_{r}$ be new variables that
don't occur in $\Pi$, we add $\alpha_{1}\wedge\cdots\wedge\alpha_{l}\rightarrow P\bar{u}$
and $P\bar{x}\rightarrow\bot$ as a conjuncts,
\item if $R_{i}\bar{u}\leftarrow\alpha_{1},\cdots,\alpha_{l}$ is a rule
of $\Pi$ and $R_{i}\neq P$, then we add the clause $\alpha_{1}\wedge\cdots\wedge\alpha_{l}\rightarrow R_{i}\bar{u}$
as a conjunct.
\end{itemize}
\end{itemize}

This completes the proof.

\end{proof}
From the above proposition and Corollary \ref{cor:D^*D^rD^*r} we
can conclude the following corollary.
\begin{corollary}
SO-HORN$^{*}\subseteq$ SO-HORN$^{r}\equiv$ SO-HORN$^{*r}$ \end{corollary}
\begin{proof}
We show that every SO-HORN$^{*}$ (SO-HORN$^{*r}$) formula is equivalent
to a SO-HORN$^{r}$ formula.

$\phi(\bar{x})$ is a SO-HORN$^{*}$ (SO-HORN$^{*r}$) formula,

$\Rightarrow$ there exists a DATALOG$^{*}$ (DATALOG$^{*r}$) formula
$\ensuremath{(\Pi,P)\bar{t}}$ such that $(\Pi,P)\bar{t}$ and $\ensuremath{\neg\phi(\bar{x})}$
are equivalent (by Proposition \ref{pro:HORN=NegDATALOG}),

$\Rightarrow$ there exists a DATALOG$^{r}$ formula $\ensuremath{(\Pi',P')\bar{t}}$
such that $(\Pi',P')\bar{t}$ and $\ensuremath{\neg\phi(\bar{x})}$
are equivalent (by Corollary \ref{cor:D^*D^rD^*r}),

$\Rightarrow$ there exists a SO-HORN$^{r}$ formula $\ensuremath{\psi(\bar{x})}$
such that $\ensuremath{\psi(\bar{x})}$ and $\phi(\bar{x})$ are equivalent
(by Proposition \ref{pro:HORN=NegDATALOG}).
\end{proof}
To compare the expressive power of the several Horn logic defined
above and DATALOG program, we define stratified versions of these
Horn logic in the following.
\begin{definition}
\label{def:SO-HORN_s}For each number $s\geq1$, SO-HORN$_{s}$ is
defined inductively by
\begin{itemize}
\item SO-HORN$_{1}=$ SO-HORN,
\item SO-HORN$_{j+1}$ is the set of formulas of the form
\[
Q_{1}R_{1}\cdots Q_{m}R_{m}\forall\bar{x}(C_{1}\wedge\cdots\wedge C_{n})
\]
 where $Q_{i}\in\{\forall,\exists\}$, $R_{1},\cdots,R_{m}$ are relation
symbols and each $C_{j}$ is an implication of the form
\[
\alpha_{1}\wedge\cdots\wedge\alpha_{l}\wedge\beta_{1}\wedge\cdots\wedge\beta_{q}\rightarrow H
\]
 where

\begin{description}
\item [{1)}] each $\alpha_{s}$ is an atomic formula $R_{i}\bar{x}$,
\item [{2)}] each $\beta_{t}$ is either a SO-HORN$_{l}$ formula $\phi(\bar{x})$
or its negation$\neg\phi(\bar{x})$ where $l\leq j$ and $R_{1},\cdots,R_{m}$
don't occur in $\phi$,
\item [{3)}] $H$ is either a atomic formula $R_{k}\overline{z}$ or the
Boolean constant $\bot$ (for false).
\end{description}
\end{itemize}

Set SO-HORN$_{\infty}=\underset{s\geq1}{\bigcup}$SO-HORN$_{s}$.

If we replace SO-HORN by SO-HORN$^{*}$, we denote this logic by SO-HORN$_{\infty}^{*}$.
And if we replace SO-HORN by SO-HORN$^{r}$, and replace condition
\textbf{1)} by
\begin{description}
\item [{1')}] each $\alpha_{s}$ is either an atomic formula $R_{i}\bar{x}$
or $\forall\bar{y}R_{i}\bar{y}\bar{z}$, we denote this logic by SO-HORN$_{\infty}^{r}$.
\end{description}
\end{definition}
\begin{proposition}
SO-HORN$_{\infty}$, SO-HORN$_{\infty}^{*}$, SO-HORN$_{\infty}^{r}$
are equivalent to S-DATALOG, S-DATALOG$^{*}$, S-DATALOG$^{r}$, respectively. \end{proposition}
\begin{proof}
The proof is similar to that of Proposition \ref{pro:HORN=NegDATALOG}.
The details are left to the reader.\end{proof}
\begin{corollary}
SO-HORN$_{\infty}\equiv$ SO-HORN$_{\infty}^{*}\equiv$ FO(BFP)
\end{corollary}
From the results above we obtain our main theorem in the following.
\begin{theorem}
On all finite structures the followings are hold:
\[
\begin{array}{ll}
 & \text{SO-HORN}\subset\text{SO-HORN}^{*}\subset\text{SO-HORN}^{r}\\
\text{and }\\
 & \text{SO-HORN}^{r}\equiv\text{SO-HORN}^{*r}\equiv\text{SO-HORN}_{\infty}^{r}\equiv\text{FO(LFP)}
\end{array}
\]

\end{theorem}
In \cite{gra} Gr\"{a}del asked that whether SO-HORN captures the class
of problems that are in P and closed under substructures. Using some
results above and a technique that encode a graph into a tree which
is exponentially larger than it, we show that if the answer is ``Yes'',
then the 3-COLORABILITY problem is in P, which implies that P=NP.
So it seems that the answer is not ``Yes'' for this question. We
ask that whether every problem that is expressible in FO(LFP) and
closed under substructures is expressible by a SO-HORN formula. This
question is equivalent to the question that whether every problem
that is expressible in DATALOG$^{r}$ and closed under extensions
is expressible by a DATALOG formula.

\section{Expressive power of Extended Horn Logic}

In \cite{zhao-1} we introduced SO-EHORN and SO-EHORN$^{r}$ logic and proved that
they capture co-NP on ordered structures. It is well known that $\Sigma_{1}^{1}$
captures NP\cite{fagin}, as a corollary, $\Pi_{1}^{1}$ captures
co-NP on all structures. So it is of interest to know whether SO-EHORN$^{r}$
captures co-NP on all structures. In this section, we show that SO-EHORN$^{r}\equiv\Pi_{1}^{1}$
which implies that SO-EHORN$^{r}$ captures co-NP on all structures.
\begin{definition}
\label{def:Second-order-Extended-Horn}Second-order Extended Horn
logic, denoted by SO-EHORN, is the set of second-order formulas of
the form
\[
\forall P_{1}\exists R_{1}\cdots\forall P_{m}\exists R_{m}\forall\bar{x}(C_{1}\wedge\cdots\wedge C_{n})
\]
 where $R_{1},\cdots,R_{m},P_{1},\cdots,P_{m}$ are relation symbols
and $C_{1},\cdots,C_{n}$ are extended Horn clauses with respect to
$R_{1},\cdots,R_{m}$, more precisely, each $C_{j}$ is an implication
of the form
\[
\alpha_{1}\wedge\cdots\wedge\alpha_{l}\wedge\beta_{1}\wedge\cdots\wedge\beta_{q}\rightarrow H
\]
 where
\begin{description}
\item [{1)}] each $\alpha_{s}$ is an atomic formula $R_{i}\bar{x}$,
\item [{2)}] each $\beta_{t}$ is either an atomic formula $Q\bar{y}$
or a negated atomic formula $\neg Q\bar{y}$ where $Q\notin\{R_{1},\cdots,R_{m}\}$,
\item [{3)}] $H$ is either an atomic formula $R_{k}\overline{z}$ or the
Boolean constant $\bot$ (for False).
\end{description}

If we replace condition \textbf{1)} by
\begin{description}
\item [{1')}] each $\alpha_{s}$ is either an atomic formula $R_{i}\overline{x}$
or $\forall\overline{y}R_{i}\overline{y}\overline{z}$,
\end{description}

then we call the logic second-order Extended revised Horn Logic, denoted
by SO-EHORN$^{r}$.

\end{definition}
By Lemma \ref{lem:clo_under_subsruc}, we know that SO-EHORN is closed substructures. So it can not captures co-NP on all structures.
But SO-EHORN$^{r}$ is more expressive, the following example shows that it can express some
co-NP complete problem.
\begin{example}
The decision problem for the set $\{\phi|\phi$ is an unsatisfiable
propositional CNF formula$\}$ is co-NP complete. We can encode $\phi$
via the structure $\mathcal{A}_{\phi}=\left\langle A,Cla,Var,P,N\right\rangle $\cite{imme-2},
where $A$ is the domain, $Cla$ and $Var$ are two unary relations,
$P$ and $N$ are two binary relations, and for any $i,j\in A,\ Cla(i),Var(j),P(i,j)$
and $N(i,j)$ means that $i$ is a clause, $j$ is a variable, variable
$j$ occurs positively in clause $i$ and variable $j$ occurs negatively
in clause $i$, respectively.

For example, the CNF formula
\[
\phi=((p_{1}\vee p_{3})\wedge(p_{2}\vee\neg p_{3})\wedge(p_{1}\vee p_{2}))
\]
 can be encoded as

\[
\begin{array}{l}
\mathcal{A}_{\phi}=\left\langle \{1,2,3\},Cla,VarP,N\right\rangle \\
Cla=\{1,2,3\}\\
Var=\{1,2,3\}\\
P=\{(1,1),(1,3),(2,2),(3,1),(3,2)\}\\
N=\{(2,3)\}
\end{array}
\]
The following formula
\[
\Phi=\forall X\exists Y\forall x\forall y\left(\begin{array}{cl}
 & \exists z\neg Y(z)\\
\wedge & (\neg Y(x)\rightarrow Cla(x))\\
\wedge & (Cla(x)\wedge Var(y)\wedge\neg Y(x)\wedge P(x,y)\rightarrow\neg X(y))\\
\wedge & (Cla(x)\wedge Var(y)\wedge\neg Y(x)\wedge N(x,y)\rightarrow X(y))
\end{array}\right)
\]
means that for any valuation $X$, where $X(i)$ holds iff variable
$i$ is true under this valuation, there exists a set of $Y$ of clauses,
such that any clause that doesn't occur in $Y$ is false under this
valuation. $\Phi$ is a SO-EHORN$^{r}$ formula and for any CNF formula
$\phi$, $\mathcal{A}_{\phi}\models\Phi$ iff $\phi$ is unsatisfiable.
\end{example}
\begin{theorem}
\label{thm:SO-EHORN^r=Pi_1}SO-EHORN$^{r}\equiv\Pi_{1}^{1}$
on all structures.\end{theorem}
\begin{proof}
We give here two different proofs of this theorem.

\#1. It is well known that every $\Sigma_{1}^{1}$ formula is equivalent
to a formula of skolem normal form\cite{gra91}
\begin{equation}
\exists P_{1}\cdots\exists P_{n}\forall x_{1}\cdots\forall x_{m}\exists y_{1}\cdots\exists y_{r}\phi\label{eq:SOE_norm}
\end{equation}
where $\phi$ is a quantifier-free CNF formula. Let $P$ be a new
$r$-ary second-order relation symbol and $\bar{z}=z_{1},\cdots,z_{r}$
be new variables that don't occur in \eqref{eq:SOE_norm}. It is easily
seen that $\exists y_{1}\cdots\exists y_{r}\phi$ is equivalent to
\[
\exists P(\exists\bar{z}P\bar{z}\wedge\forall\bar{y}(P\bar{y}\rightarrow\phi))
\]
Then \eqref{eq:SOE_norm} is equivalent to
\[
\exists P_{1}\cdots\exists P_{n}\forall\bar{x}\exists P\forall\bar{y}(\exists\bar{z}P\bar{z}\wedge(P\bar{y}\rightarrow\phi))
\]
Repeated application of Lemma \ref{lemma:UviFOafterExisSO} shows that
the above formula is equivalent to
\begin{equation}
\exists P_{1}\cdots\exists P_{n}\exists P'\forall\bar{x}\forall\bar{y}(\exists\bar{z}P'\bar{x}\bar{z}\wedge(P'\bar{x}\bar{y}\rightarrow\phi))\label{eq:SO_nom2}
\end{equation}
where $P'$ is an $m+r$-ary second-order relation symbol. The left
is the same as that of Theorem 7 in \cite{zhao-1}, just replace the
formula $\Phi$ by the formula \eqref{eq:SO_nom2} above.

\#2. By \eqref{eq:SOE_norm} we know that every $\Pi_{1}^{1}$ formula
is equivalent to a formula of the form
\begin{equation}
\forall P_{1}\cdots\forall P_{n}\exists x_{1}\cdots\exists x_{m}\forall y_{1}\cdots\forall y_{r}\phi'\label{eq:SO_nom3}
\end{equation}
where $\phi'$ is a quantifier-free formula. With no loss of generality,
suppose $\phi'$ is in CNF form. Let $R$ be a new $m$-ary second-order
relation symbol and $\bar{z}=z_{1},\cdots,z_{m}$ be new variables
that don't occur in \eqref{eq:SO_nom3}, then $\exists\bar{x}\forall\bar{y}\phi'$
is equivalent to $\exists R(\exists\bar{z}\neg R\bar{z}\wedge\forall\bar{x}(\neg R\bar{x}\rightarrow\forall\bar{y}\phi'))$.
So \eqref{eq:SO_nom3} is equivalent to
\begin{equation}
\forall P_{1}\cdots\forall P_{n}\exists R\forall\bar{x}\forall\bar{y}(\exists\bar{z}\neg R\bar{z}\wedge(\neg R\bar{x}\rightarrow\phi'))\label{eq:SO_nom4}
\end{equation}
It is easily seen that \eqref{eq:SO_nom4} is equivalent to a SO-EHORN$^{r}$
formula.\end{proof}
\begin{corollary}
SO-EHORN$^{r}$ captures co-NP on all structures.
\end{corollary}

\begin{corollary}
SO-EHORN$^{r}$ collapses to $\Pi_{2}^{1}$-EHORN$^{r}$ on all structures.\end{corollary}
\begin{proof}
In the second proof of Theorem \eqref{thm:SO-EHORN^r=Pi_1},
we have proved more, namely that every $\Pi_{1}^{1}$ formula is equivalent
to a $\Pi_{2}^{1}$-EHORN$^{r}$ formula which implies that SO-EHORN$^{r}\subseteq\Pi_{2}^{1}$-EHORN$^{r}$.
\end{proof}

\section{Conclusion}

In this paper, we showed that, on all finite structures, SO-HORN$^{r}$,
SO-HORN$^{*r}$, FO(LFP) coincide with each other and SO-HORN$^{*}$
is a proper sublogic of SO-HORN$^{r}$. We introduced the notions
of DATALOG$^{*}$ program and DATALOG$^{r}$ program and their stratified
versions, S-DATALOG$^{*}$ program and S-DATALOG$^{r}$ program. We
proved that DATALOG$^{r}$ and S-DATALOG$^{r}$ are equivalent and
DATALOG$^{*}$ is a proper sublogic of DATALOG$^{r}$.
We also showed that SO-EHORN$^{r}$ which is an extended version of
SO-HORN captures co-NP on all structures. We prove this this result
in two versions. In the first version we improved the proof in \cite{zhao-1}
using the skolem normal forms of $\Sigma_{1}^{1}$ formulas. In the
second version, we gave a straightforward proof using the normal form.



\vspace{2ex}
\begin{thebibliography}{99}

\bibitem[1]{fl90} A. Fl\"{o}gel, M. Karpinski, and H. Kleine B\"{u}ning.
Subclasses of Quantified Boolean Formulas. In Lecture Notes in Computer
Science, 553: 145-155, Springer, 1990.

\bibitem[2]{fl93} A. Fl\"{o}gel. Resolution f\"{u}r quantifizierte Bool'sche
Formeln. Dissertation, University Paderborn ,1993.

\bibitem[3]{chanda85}A. Chandra, D. Harel. Horn Clause Queries and Generalizations.
J. Logic Programming 1 (1985), 1\textendash{}15.

\bibitem[4]{gra91}Erich Gr\"{a}del. The expressive power of second order
Horn logic. Proceedings of the 8th annual symposium on Theoretical
aspects of computer science, p.466-477, February 1991, Hamburg.

\bibitem[5]{gra} E. Gr\"{a}del. Capturing Complexity Classes by Fragments
of Second-order Logic. Theoretical Computer Science, 101: 35-57, 1992.

\bibitem[6]{Gradel}Erich Gr\"{a}del , P. G. Kolaitis, L. Libkin , M.
Marx , J. Spencer , Moshe Y. Vardi , Y. Venema , Scott Weinstein.
Finite Model Theory and Its Applications. Springer-Verlag New York,
Inc., Secaucus, NJ, 2005.

\bibitem[7]{Evgeny}Evgeny Dantsin, Thomas Eiter, Georg Gottlob, Andrei
Voronkov. Complexity and Expressive Power of Logic Programming. ACM
Computing Surveys, Vol. 33, No. 3, September 2001, pp. 374\textendash{}425.

\bibitem[8]{ef} H.-D. Ebbinghaus, J. Flum. Finite Model Theory. Springer,
1999.

\bibitem[9]{kble} H. Kleine B\"{u}ning, T. Lettmann. Propositional
Logic: Deduction and Algorithms. Cambridge University Press, 1999.

\bibitem[10]{Hans}Hans Kleine B\"{u}ning, Xishun Zhao. On Models for
Quantified Boolean Formulas. In Proc. of SAT'04, 2004.

\bibitem[11]{kbzh} H. Kleine B\"{u}ning, Xishun Zhao. Computational Complexity
of Quantified Boolean Formulas with Fixed Maximal Deficiency. Theoretical
Computer Science, 407: 448-457, 2008.

\bibitem[12]{ajtai}Miklos Ajtai, Yuri Gurevich. Datalog vs first-order
logic. Journal of Computer and System Sciences, Volume 49, Issue 3,
December 1994, Pages 562-588.

\bibitem[13]{imme}N. Immerman. Relational Queries Computable in Polynomial
Time. Inf. and Control 68 (1986), 86\textendash{}104.

\bibitem[14]{imme-1}N. Immerman. Languages that Capture Complexity
Classes. SIAM J. Comput. 16 (1987), 760\textendash{}778.

\bibitem[15]{imme-2}N. Immerman. Descriptive Complexity. Springer-Verlag,
New York, 1999.

\bibitem[16]{pk}P. Kolaitis. The Expressive Power of Stratified Logic
Programs. Information and Computation archive Volume 90 , Issue 1,
Pages: 50 - 66,1991.

\bibitem[17]{fagin}R. Fagin. Generalized First-Order Spectra and
Polynomial-Time Recognizable Sets. SIAM-AMS Proc. 7 (1974), 43\textendash{}73.

\bibitem[18]{abiteboul1}S. Abiteboul , V. Vianu. Fixpoint extensions
of first-order logic and datalog-like languages. Proceedings of the
Fourth Annual Symposium on Logic in computer science, p.71-79, June
1989, California, United States

\bibitem[19]{abiteboul2}S. Abiteboul , V. Vianu. DATALOG extensions
for database queries and updates. Journal of Computer and System Sciences
43,62-124,1991.

\bibitem[20]{abiteboul3}S. Abiteboul , R. Hull, V. Vianu. Foundations
of Data Bases. Addison Wesley, 1995.

\bibitem[21]{zhao-1}Shiguang Feng, Xishun Zhao. Complexity and Expressive
Power of Second-Order Extended Horn Logic. Mathematical Logic Quarterly,
under review.

\bibitem[22]{zhao}Xishun Zhao, Hans Kleine B\"{u}ning. Model-Equivalent
Reductions. Lecture Notes in Computer Science, Volume 3569/2005, 355-370,
Springer, 2005.

\bibitem[23]{gurevich}Y. Gurevich. Logic and the Challenge of Computer
Science. Computer Science Press (1988), 1\textendash{}57.


\end{thebibliography}
\end{document}